\documentclass{patmorin}
\usepackage[utf8]{inputenc}

\usepackage{amsthm,amsmath,amssymb}
\usepackage{graphicx}
\usepackage{subcaption}
\usepackage{pat}
\usepackage[letterpaper]{hyperref}
\usepackage[table,dvipsnames]{xcolor}
\usepackage{thm-restate}
\usepackage{float}

\usepackage{paralist}

\definecolor{linkblue}{named}{Blue}
\hypersetup{colorlinks=true, linkcolor=linkblue,  anchorcolor=linkblue,
citecolor=linkblue, filecolor=linkblue, menucolor=linkblue,
urlcolor=linkblue}
\usepackage{wasysym}

\DeclareMathOperator{\sn}{sn}

\DeclareMathOperator{\qn}{qn}
\DeclareMathOperator{\tn}{tn}
\DeclareMathOperator{\tr}{tn}
\DeclareMathOperator{\pw}{pw}
\DeclareMathOperator{\tw}{tw}
\DeclareMathOperator{\lpw}{lpw}
\DeclareMathOperator{\ltw}{ltw}

\title{\MakeUppercase{Two Results on Layered Pathwidth and Linear Layouts}%
  \thanks{This research was partly funded by NSERC and the Ontario Ministry of Economic Development, Job Creation and Trade (formerly the Ministry of Research and Innovation)}}

\author{Vida Dujmović%
  \thanks{Deparment of Computer Science and Electrical Engineering, University of Ottawa}\qquad
  Pat Morin%
  \thanks{School of Computer Science, Carleton University}\qquad
  Céline Yelle\footnotemark[2]}

\date{}

\begin{document}
%
%
%

%
\maketitle              

\begin{abstract}
  Layered pathwidth is a new graph parameter studied by Bannister \etal\ (2015). In this paper we present two new results relating layered pathwidth to two types of linear layouts.  Our first result shows that, for any graph $G$, the stack number of $G$ is at most four times the layered pathwidth of $G$.   Our second result shows that any graph $G$ with track number at most three has layered pathwidth at most four.  The first result complements a result of Dujmović and Frati (2018) relating layered treewidth and stack number.  The second result solves an open problem posed by Bannister \etal\ (2015).
  %
\end{abstract}


\section{Introduction}
\pagenumbering{arabic}

The treewidth and pathwidth of a graph are important tools in
structural and algorithmic graph theory. Layered treewidth and layered
$H$-partitions are recently developed tools that generalize
treewidth. These tools played a critical role in recent breakthroughs on a
number of longstanding problems on planar graphs and their generalizations, including the queue number of planar graphs \cite{dujmovic.joret.ea:planar}, the nonrepetitive chromatic number of planar graphs \cite{dujmovic.esperet.ea:planar}, centered colourings of planar graphs \cite{debski.felsner.ea:improved}, and adjacency labelling schemes for planar graphs \cite{bonamy.gavoille.ea:shorter,dujmovic.esperet.ea:adjacency}.


Motivated
by the versatility and utility of layered treewidth, Bannister
\etal\ \cite{DBLP:conf/gd/BannisterDDEW16,bannister2018track}
introduced  layered pathwidth, which generalizes pathwidth in the same way
that layered treewidth generalizes treewidth.  The goal of this article is to fill the gaps in our knowledge about
the relationship between layered pathwidth and the following well studied
linear graph layouts: queue-layouts, stack-layouts and track
layouts.  We begin by defining all these terms.

\subsection{Layered Treewidth and Pathwidth}

A {\em tree decomposition} of a graph $G$ is given by a tree $T$ whose
nodes index a collection of sets $B_1,\ldots,B_p\subseteq V(G)$ called
\emph{bags} such that (1) for each $v\in V(G)$, the set $T[v]$ of bags that contain $v$ induces a
     non-empty (connected) subtree in $T$; and (2)~for each edge $vw\in E(G)$, there is some bag that contains both $v$ and $w$. If $T$ is a path, the resulting decomposition is a called a \emph{path decomposition}. The \emph{width} of a tree (path)
   decomposition is the size of its largest bag.  The \emph{treewidth}
   (\emph{pathwidth}) of $G$, denoted $\tw(G)$ ($\pw(G)$), is the minimum width of any tree (path) decomposition of $G$ minus $1$.

   A \emph{layering} of $G$ is a mapping $\ell:V(G)\to\Z$ with the
property that $vw\in E(G)$ implies $|\ell(v)-\ell(w)|\le 1$. One
can also think of a layering as a partition of $G$'s vertices into
sets indexed by integers, where $L_i=\{v\in V(G)
: \ell(v)=i\}$ is called a  \emph{layer}.  A \emph{layered tree (path) decomposition} of $G$ consists
of a layering $\ell$ and a tree (path) decomposition with bags $B_1,\ldots,B_p$ of $G$.
The \emph{(layered) width} of a layered tree (path) decomposition is the maximum
size of the intersection of a bag and a layer, i.e., $\max\{|L_i\cap
B_j|:i\in\Z,\, j\in\{1,\ldots,p\}\}$.  The \emph{layered treewidth (pathwidth)} of
$G$, denoted $\ltw(G)$ ($\lpw(G)$) is the smallest (layered) width of any layered
tree (path) decomposition of $G$.

\figref{grid} shows a layered path decomposition of a grid with diagonals, $G$, in which each pair of consecutive columns is contained in a common bag $B_i$ each row is contained in a single layer $L_j:=\{x\in V(G):\ell(x)=j\}$.  This decomposition shows that $G$ has layered pathwidth 2 since each bag $B_i$ contains two elements per layer $L_j$.
\begin{figure}
    \centering{\includegraphics{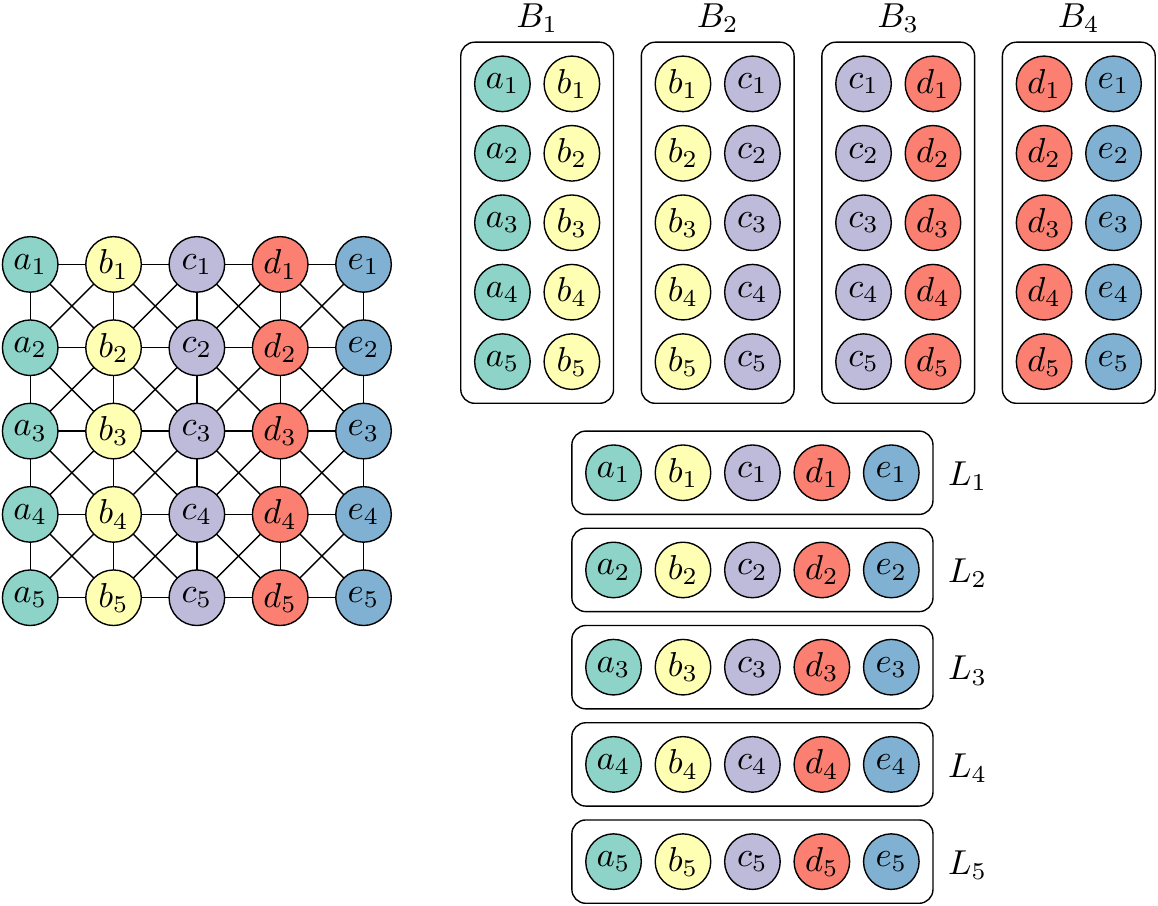}}
    \caption{A width-2 layered path decomposition of the $5\times 5$ grid graph $G$.}
    \figlabel{grid}
\end{figure}


Note that while layered pathwidth is at most pathwidth, pathwidth is
not bounded\footnote{We say that an integer-valued graph parameter $\mu$ is \emph{bounded} by some other integer-valued graph parameter $\nu$ if there exists a function $f:\N\to\N$ such that $\mu(G)\le f(\nu(G))$ for every graph $G$.} by layered pathwidth. There are graphs---for example the $n \times
n$ planar grid---that have unbounded pathwidth and bounded layered
pathwidth. Thus upper bounds proved in terms of layered pathwidth are
quantitatively stronger than those proved in terms of pathwidth. In addition, while having pathwidth at most $k$ is a minor-closed property,\footnote{A graph $H$ is a minor of a graph $G$ if a graph isomorphic to $H$ can be obtained from a
subgraph of $G$ by contracting edges. A class $G$ of graphs is minor-closed if for every graph $G\in\mathcal{G}$, every minor of $G$
is in $\mathcal{G}$. A minor-closed class is proper if it is not the class of all graphs.} having layered pathwidth at most $k$ is not.  For
example, the $2 \times n \times n$ grid graph has layered pathwidth at most $3$ but
it has  $K_n$ as a minor, and thus it has a minor of unbounded layered pathwidth. (Analogous statements hold for layered treewidth)



\subsection{Linear Layouts}

After introducing layered path decompositions, Bannister \etal\ \cite{DBLP:conf/gd/BannisterDDEW16,bannister2018track} set
out to understand the relationship between track/queue/stack number and
layered pathwidth.

A \emph{$t$-track layout} of a graph
$G$ is a partition of $V(G)$ into $t$ ordered independent sets $T_1,\ldots,T_t$ (with a total order $\prec_i$ for each $T_i$, $i\in\{1,\ldots,t\}$) with no X-crossings. Here an \emph{X-crossing} is a pair of edges $vw$ and $xy$ such that, for some $i,j\in\{1,\ldots,t\}$, $v,x\in T_i$ with $v\prec_i x$ and $w,y\in T_j$ with $y\prec_j w$. The minimum number of tracks in any $t$-track layout of $G$ is called the \emph{track number} of $G$ and is denoted as $\tr(G)$. A \emph{$t$-track graph} is a graph that has a $t$-track layout.

A \emph{stack (queue) layout} of a graph $G$ consists of a total order
$\sigma$ of $V(G)$ and a partition of $E(G)$ into  sets, called
\emph{stacks} (\emph{queues}), such that no two edges in the same
stack (queue) \emph{cross}; that is, there are no edges
$vw$ and $xy$ in a single stack with $v\prec_\sigma x\prec_\sigma w\prec_\sigma y$
(\emph{nest}; there are no edges $vw$ and $xy$ in a single queue with $v\prec_\sigma
x\prec_\sigma y\prec_\sigma w$.).  The minimum number of stacks (queues) in a
stack (queue) layout of $G$ is the \emph{stack number} (the
\emph{queue number}) of $G$ and is denoted as $\sn(G)$ ($\qn(G)$). A stack layout is also called a {\em book embedding} and stack number is also called {\em book thickness} and {\em page number}. An \emph{$s$-stack graph} (\emph{$q$-queue graph}) is a graph that has a stack (queue) layout with at most $s$ stacks ($q$ queues).

\subsection{Summary of (Old and New) Results}
A summary of known and new results on these rich relationships between layered pathwidth, queue number, stack number, and track number are outlined in \tabref{summary}.  The first two rows show (the older results) that track number and queue number are \emph{tied}; each is bounded by some function of the other \cite{dmw05,dpw04}.

\begin{table}
  \begin{center}
    \begin{tabular}{|l@{\hspace{1em}}l|}\hline
      \multicolumn{2}{|l|}{Queue-Number versus Track-Number} \\ \hline
      $\qn(G) \le \tn(G)-1$ & \cite[Theorem~2.6]{dmw05} \\
      $\tn(G) \le 2^{O(\qn(G)^2)}$ & \cite[Theorem~8]{dpw04} \\ \hline
      \multicolumn{2}{l}{} \\
      %
      %
      \hline
      \multicolumn{2}{|l|}{Queue-Number versus Layered Pathwidth} \\ \hline
      $\qn(G) \le 3\lpw(G)-1$ & \cite[Theorem~2.6]{dmw05}\cite[Lemma~9]{bannister2018track} \\
      $\qn(G) = 1 \Rightarrow \lpw(G)\le 2$ &
      \cite[Theorem~3.2]{HR-SJC92}\cite[Corollary~7]{bannister2018track} \\
      $\exists G : \qn(G)=2,\, \lpw(G)=\Omega((n/\log n)$
      & \cite[Theorem~1.4]{dsw16} \\ \hline
      \multicolumn{2}{l}{} \\
      \hline
      \multicolumn{2}{|l|}{Stack-Number versus Layered Pathwidth} \\ \hline
      $\sn(G) \le 1 \Rightarrow \lpw(G) \le 2$
      & \cite[Corollary~16]{bannister2018track} \\
      $\exists G: \sn(G)=2,\, \lpw(G) = \Omega(\log n)$
      & ($G$ is a binary tree plus an apex vertex)\\
      $\exists G: \sn(G)=3,\, \lpw(G) = \Omega(n/\log n)$
      & \cite[Theorem~1.5]{dsw16} \\
      $\sn(G) \le 4\lpw(G)$ & \textbf{\thmref{stacknumber}} \\ \hline
      \multicolumn{2}{l}{} \\
      \hline
      \multicolumn{2}{|l|}{Track-Number versus Layered Pathwidth} \\ \hline
      $\tn(G) \le 3\lpw(G)$ & \cite[Lemma~9]{bannister2018track} \\
      $\tn(G) = 1 \Rightarrow \lpw(G)= 1$ & ($G$ has no edges) \\
      $\tn(G) \le 2\Rightarrow \lpw(G) \le 2$ & ($G$ is a forest of caterpillars) \\
      $\tn(G) \le 3\Rightarrow \lpw(G) \le 4$ & \textbf{\thmref{main}} \\
      $\exists G : \tn(G)=4,\, \lpw(G)=\Omega((n/\log n)$
      & \cite[Theorem~1.5]{dsw16} \\ \hline
    \end{tabular}
  \end{center}
  \caption{Relationships between track number, queue number, stack number, and layered pathwidth.}
  \tablabel{summary}
\end{table}

The next group of rows relates queue number and layered pathwidth. Queue number is bounded by layered pathwidth \cite{dmw05}.  Graphs with queue number 1 are arched-levelled planar graphs\footnote{An \emph{arched leveled} planar embedding is one in which the vertices are placed on parallel lines (levels) and each edge either connects vertices on two consecutive levels or forms an arch that connects two vertices on the same level by looping around all previous levels.  A graph is arched levelled if it has an arched levelled planar embedding.} and have layered pathwidth at most 2.\footnote{Theorem~6 in \cite{bannister2018track} can easily be modified to
prove that arched levelled planar graphs have layered
pathwidth at most 2. That is achieved by adding the leftmost vertex of each level to each bag of the path decomposition.}  However, there are graphs with queue number 2 that are expanders; these graphs have pathwidth $\Omega(n)$ and diameter $O(\log n)$, so their layered pathwidth is $\Omega(n/\log n)$ \cite{dsw16}.  Thus, layered pathwidth is not bounded by queue number.

The next group of rows examines the relationship between stack number and layered pathwidth.  Graphs of stack number at most 1 are exactly the outerplanar graphs, which have layered-pathwidth at most 2 \cite{bannister2018track}.  On the other hand, there are graphs of stack number 2 that have unbounded layered pathwidth \cite{dsw16}.  Thus, in general, layered pathwidth is not bounded by stack number.  Our first result, \thmref{stacknumber}, shows that stack number is nevertheless bounded by layered pathwidth.

\begin{thm}\thmlabel{stacknumber}
 For every graph $G$, $\sn(G)\le 4 \lpw(G)$.
\end{thm}

The final group of rows relates track number and layered pathwidth.  Track number is bounded by layered pathwidth \cite{bannister2018track}.  Layered pathwidth is bounded by track number when the track number is 1, or 2, but is not bounded by track number when the track number is 4 or more \cite{bannister2018track}.  The question of what happens for track number 3 is stated as an open problem by Bannister \etal\ \cite{bannister2018track}, who solved the special case when $G$ is bipartite and has track number 3.  Our \thmref{main} solves this problem completely by showing that graphs with track number at most 3 have layered pathwidth at most 4.

Note that minor-closed classes that have bounded layered pathwidth have been characterized (as classes of graphs that exclude an apex tree\footnote{A graph $G$ is an {\em apex tree} if it has a vertex $v$ such that $G-v$ is a forest. } as a minor) \cite{DBLP:journals/corr/abs-1810-08314}. However, this result could not have been used to prove \thmref{main} since the family of 3-track graphs is not closed under taking minors.\footnote{To see this, start with an $ n\, \times\, n$ planar grid. Every planar grid has a 3-track layout. However, for large enough $n$, one can contract/delete edges on this grid graph such that the result is a series-parallel graph that does not have a 3-track layout (in particular the series-parallel graph from Theorem 18 in \cite{bannister2018track}).}

\begin{thm}\thmlabel{main}
 Every graph $G$ that has $\tr(G)\le 3$, has $\lpw(G)\le 4$.
\end{thm}

\section{Proof of \thmref{stacknumber}}

Let $G$ be any graph, let $B=B_1,\ldots,B_p$ be a path decomposition of $G$, and $\ell:V(G)\to\Z$ be a layering that obtains so that $B$ has layered width $\lpw(G)$ with respect to the layering $\ell$.

We may assume that $B$ is \emph{left-normal} in the sense that, for every distinct pair $v,w\in V(G)$, $\min\{i\in\Z : v\in B_i\} \neq \min\{i\in\Z:w\in B_i\}$. It is straightforward to verify that any path decomposition can be made left-normal without increasing the layered width of the decomposition.
We use the notation $v\prec_B w$ if $\min\{i\in\Z : v\in B_i\} < \min\{i\in\Z:w\in B_i\}$.  Since $B$ is left-normal, $\prec_B$ is a total order.

For any edge $vw$ with $v\prec_B w$ we call $v$ the \emph{left endpoint} of the edge and $w$ the \emph{right endpoint}.  We use the convention of writing any edge with endpoints $v$ and $w$ as $vw$ where $v$ is the left endpoint and $w$ is the right endpoint.  With these definitions and conventions in hand, we are ready to proceed.

\begin{proof}[Proof of \thmref{stacknumber}]
  This proof is illustrated in \figref{stack-layout}.
  Consider a left-normal path decomposition $B=B_1,\ldots,B_p$ of $G$ and a layering $\ell:V(G)\to\Z$ such that $B$ has layered width at most $k$ with respect to $\ell$.  Our goal is to construct a stack layout of $G$ using $4k$ stacks.
  \begin{figure}
      \centering{\includegraphics[width=.98\textwidth]{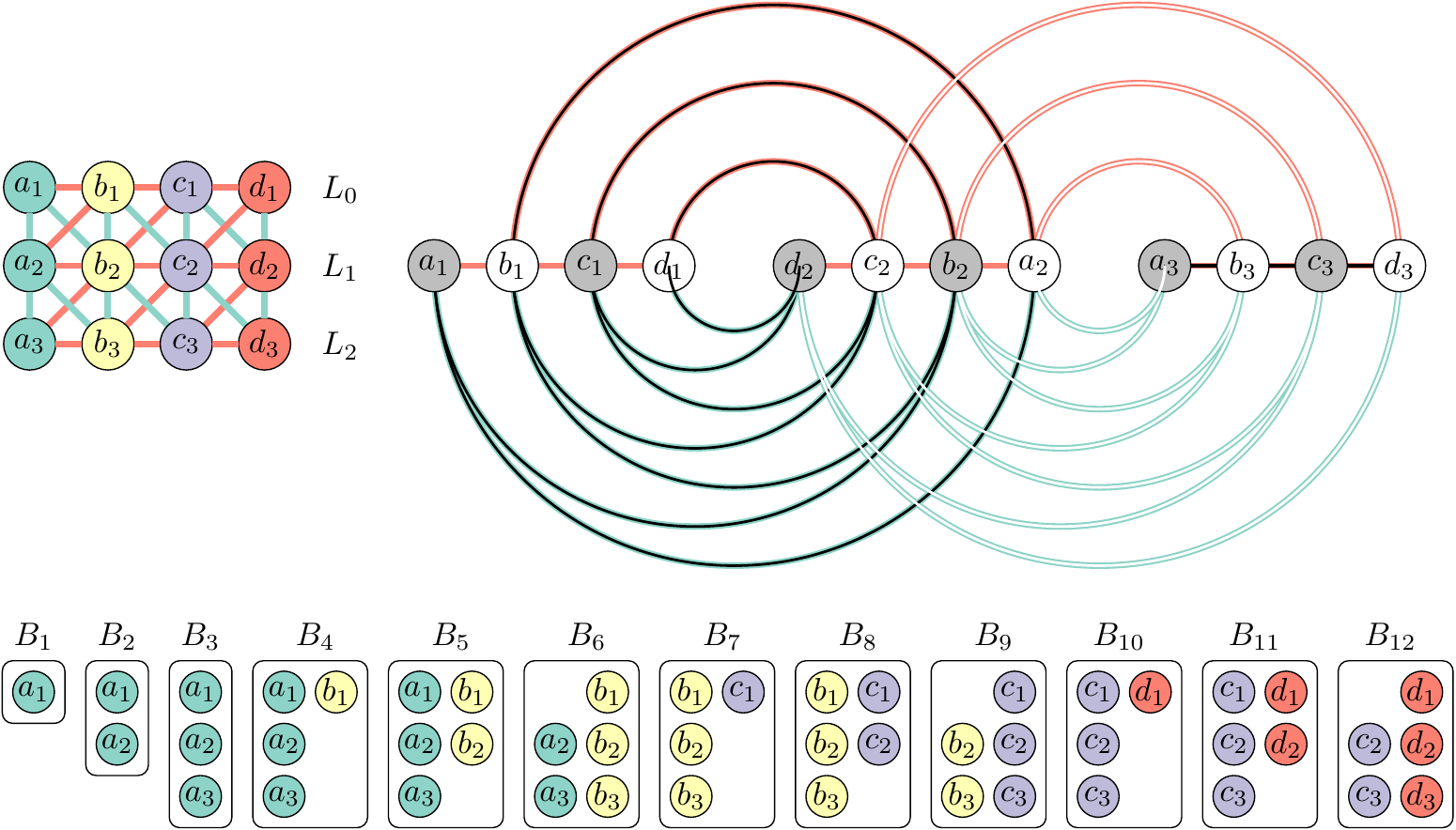}}
      \caption{Colourings used in the proof of \thmref{stacknumber}.}
      \figlabel{stack-layout}
  \end{figure}

  We first construct a total ordering $\prec_\sigma$ on the vertices of $G$, as follows:
  \begin{compactenum}[(Property 1)]
    \item If $v\in L_i$ and $w\in L_j$ with $i < j$, then $v\prec_\sigma w$.
    \item If $v,w\in L_i$ with $v \prec_B w$ then
    \begin{compactenum}[(a)]
      \item $v\prec_\sigma w$ if $i$ is even; or
      \item $w\prec_\sigma v$ if $i$ is odd.
    \end{compactenum}
  \end{compactenum}

  Next we define a colouring $\varphi:E(G)\to\{0,1\}\times\{0,1\}\times\{1,\ldots,k\}$ that determines a partition of the $E(G)$ into stacks.  We begin with a (greedy) vertex $k$-colouring $\varphi:V\to \lbrace 1,\ldots,k\rbrace$ so that, for any $i,j\in\N$, no two vertices in $B_i\cap L_j$ are assigned the same colour. This is easily done since, for each $j\in\Z$, the path decomposition $\langle B_i\cap L_j : i\in \Z\rangle$ has bags of size at most $k$.

  We say that an edge $vw$ has \emph{positive slope} if $\ell(v)=\ell(w)+1$ and has \emph{non-positive slope} otherwise.  We colour the edge $vw$ with the colour $\varphi(vw)=(a,b,c)$ where $a=\ell(v)\bmod 2$, $b\in\{0,1\}$ is a bit indicating whether $vw$ has positive ($b=1$) or non-positive ($b=0$) slope, and $c$ is the colour $\varphi(v)$ of the left endpoint $v$.  This clearly uses only $2\times2\times k=4k$ colours so all that remains is to show that
  $\sigma$ and the partition $P=\{\{vw\in E(G):\varphi(vw)=(a,b,c)\}:(a,b,c)\in \{0,1\}\times\{0,1\}\times\{1,\ldots,k\}\}$  is indeed a stack layout.

  Consider any two distinct edges $vw,xy\in E(G)$ (whose left endpoints are $v$ and $x$, respectively).  First observe that, if $\ell(v)\equiv \ell(x)\pmod 2$ then either $\ell(v)=\ell(x)$ or $\ell(v)-\ell(x)\ge 2$. In the latter case, the only way in which $vw$ and $xy$ can cross with respect to $\prec_\sigma$ is if $\ell(v)+b = \ell(y)=\ell(w) = \ell(x)-b$ for some $b\in\{-1,1\}$.  However, in this case, $vw$ has positive slope and $xy$ has non-positive slope, or vice-versa, so $\varphi(vw)$ and $\varphi(xy)$ differ in their second component.  For example, we could have $\ell(x)=2$, $\ell(v)=0$, and $b=1$, in which case $\ell(y)=\ell(w)=1$, in which case $xy$ has positive slope and $vw$ has non-positive (in face, negative) slope.

  Therefore, we only need to consider pairs of edges $xy$ and $vw$ where $\ell(v)=\ell(x)=i$. We assume, without loss of generality that $i$ is even
  and that $v\prec_\sigma x$.  With these assumptions, there are only three cases in which $vw$ and $xy$ can cross:
  \begin{enumerate}
    \item $v\prec_\sigma x\prec_\sigma w\prec_\sigma y$.

    Since $\ell(v)=\ell(x)=i$ is even and $v\prec_\sigma x$, we have $v\prec_B x$  (by Property~2a) and $\ell(w)\ge i$ (by Property~1).  If $\ell(w)=i$, then $v\prec_B x\prec_B w$ (by Property~2a), so $v$ and $x$ both appear in some bag $B_j$ and $\varphi(v)\neq\varphi(x)$, so $\varphi(vw)$ and $\varphi(xy)$ differ in their third component.  If $\ell(w)=i+1$, then $w\prec_\sigma y$ implies that $\ell(y)\ge\ell(w)$ (by Property~1), which implies $\ell(y)=\ell(w)=i+1$, so $y\prec_B w$ (by Property~2b).  We now have $v\prec_B x\prec_B y\prec_B w$ so $v$ and $x$ appear in a common bag $B_j$ and $\varphi(vw)$ and $\varphi(xy)$ differ in their third component.

    \item $v\prec_\sigma y\prec_\sigma w\prec_\sigma x$.  Since $v\prec_\sigma y$, $\ell(y)\ge \ell(v)=i$ (by Property~1).  Similarly, since $y\prec_\sigma x$, $\ell(y)\le\ell(x)=i$ (by Property~1).  Therefore, $\ell(y)=i$, so $y\prec_B x$ (by Property~2a).  This is not possible since, by definition, $x$ is the left endpoint of $xy$.

    \item $y\prec_\sigma v\prec_\sigma x\prec_\sigma w$.  Since $y\prec_\sigma x$, $ell(y)\le \ell(x)=i$ (by Property~1).  If $\ell(y)=i$ then, since $i$ is even, $y\prec_B x$ (by Property~2a) so it must be the case that $\ell(y)=i-1$. Since $x$ is the left endpoint of $xy$, $xy$ therefore has positive slope.

    Since $v\prec_\sigma w$ and $\ell(v)=i$ is even, we have $\ell(v)\le\ell(w)$ (by Property~2a). Since $v$ is the left endpoint of $vw$, $vw$ therefore has non-positive slope.  Therefore $\varphi(vw)$ and $\varphi(xy)$ differ in their second component.
  \end{enumerate}
  Therefore, for any pair of edges $vw,xy\in E(G)$ that cross, $\varphi(vw)\neq\varphi(xy)$, so the partition $P$ is a partition of $V(G)$ into $4k$ stacks with respect to $\prec_\sigma$, as required.
\end{proof}

\section{Proof of \thmref{main}}

Let $G$ be an edge-maximal $n$-vertex graph with $\tr(G)=3$.
Here, $G$ is \emph{edge-maximal} if adding any edge $e$ increases the track number to four or more. It is
helpful to recall that $G$ is a planar graph that has a straight-line
crossing-free drawing with the vertices of the track $T_1$ placed on the positive x-axis,
the vertices of the track $T_2$ placed on the positive y-axis and the vertices of
the track $T_3$ placed on the ray $\{(a,a):a<0\}$. See \figref{planar-view}.

\begin{figure}
  \begin{center}
     \includegraphics{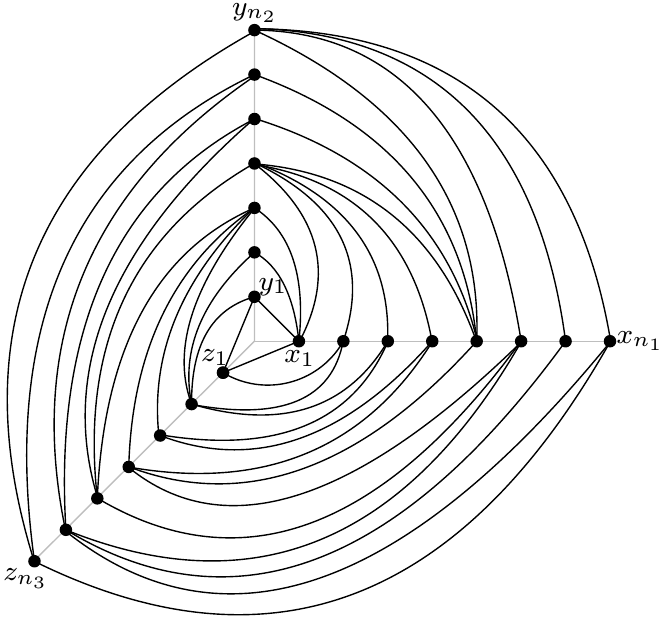}
  \end{center}
  \caption{The standard planar embedding of a 3-track graph.}
  \figlabel{planar-view}
\end{figure}

It will be easier to prove \thmref{main} for a weaker notion
of layering.  An \emph{$s$-weak layering} of $G$ is a mapping
$\ell:V(G)\to\Z$ with the property that, for every $vw\in E(G)$,
$|\ell(v)-\ell(w)|\le s$.  The sets $L_i=\{v\in V(G): \ell(v)=i\}$
are called \emph{layers}.  The terms \emph{$s$-weak layered path
decomposition} and \emph{$s$-weak layered pathwidth} of $G$, denoted
$\lpw_s(G)$, are defined the same way as layered path decompositions
and layered pathwidth, but with respect to $s$-weak layerings of $G$. The following result is easy (and well-known in the context of layered treewidth):

\begin{lem}\lemlabel{weak}
  For any $s\in\N$, $\lpw(G) \le s\cdot\lpw_s(G)$.
\end{lem}

\begin{proof}
    By definition there exists an $s$-weak layering $\ell:V(G)\to N$ that defines layers $L_i=\{v\in V(G):\ell(v)=i\}$, $i\in\Z$, and a path decomposition $B_1,\ldots,B_p$ of $G$ such that $|L_i\cap B_j|\le \lpw_s(G)$ for each $i\in\Z$ and $j\in\{1,\ldots,p\}$.

    Let $\ell'(v)=\lfloor \ell(v)/s\rfloor$ for each $v\in V(G)$.  For any edge $vw\in E(G)$, $\ell(v)-\ell(w)\le s$, so $\ell'(v)-\ell'(w)\le 1$, so $\ell'$ is a layering of $G$.  For each $i\in\Z$, the layer $L_i'=\{v\in V(G):\ell'(v)=i\}=\bigcup_{t=0}^{s-1} L_{is+t}$.  Therefore, for each $i\in \Z$ and $j\in\{1,\ldots,p\}$, $L'_i\cap B_j\le\sum_{t=0}^{s-1} |L_{is+t}\cap B_j| \le s\cdot\lpw_s(G)$.  Thus, $\ell'$ and $B_1,\ldots,B_p$ are a layered path decomposition of $G$ of layered width at most $s\cdot\lpw_s(G)$, so $\lpw(G)\le s\cdot\lpw_s(G)$.
\end{proof}

Let $T_1,T_2,T_3$ be a 3-track layout of $G$ with
$T_1=\{x_1,\ldots,x_{n_1}\}$, $T_2=\{y_1,\ldots,y_{n_2}\}$, and
$T_3=\{z_1,\ldots,z_{n_3}\}$ and the total orders $\prec_1,\prec_2,\prec_3$
are implicit so that, for example $z_i\prec_3 z_j$ if and only if $1\le i<j\le n_3$.
In terms of \figref{planar-view}, this means that $x_1,y_1,z_1$ form
the triangular face containing the origin and $x_{n_1},y_{n_2},z_{n_3}$
form the cycle on the boundary of the outer face.
From this point onward, all track indices are implicitly taken ``modulo 3''
so that for any integer $j$, $T_j$ refers to the track $T_{j'}$ with
index $j'=((j-1)\bmod 3)+1$.

The following observation follows from the fact that $G$ is edge-maximal.
\begin{obs}\obslabel{silly}
  For any two vertices of $G$ on distinct tracks, say $x_i$ and $y_j$, at least
  one of the following conditions is satisfied (see \figref{3-cases}):
  \begin{enumerate}
    \item $x_{i}y_{j}\in E(G)$; or
    \item there exists $x_{i'}y_{j'} \in E(G)$ with $i'>i$ and $j'<j$; or
    \item there exists $x_{i}y_{j'}\in E(G)$ with $j'>j$.
  \end{enumerate}
\end{obs}

\begin{figure}
   \begin{center}
     \begin{tabular}{ccc}
       \includegraphics{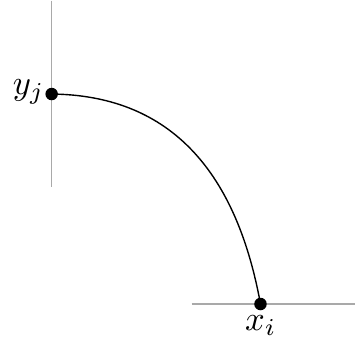} &
       \includegraphics{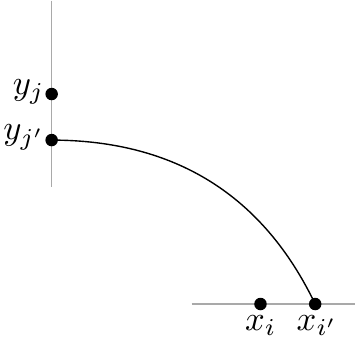} &
       \includegraphics{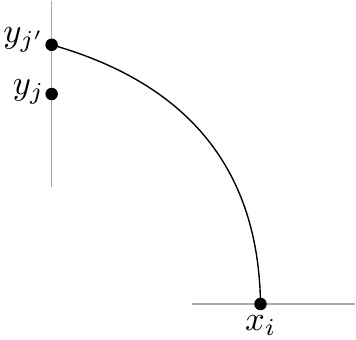} \\
       (1) & (2) & (3)
     \end{tabular}
   \end{center}
   \caption{The three cases in \obsref{silly}.}
   \figlabel{3-cases}
\end{figure}

\thmref{main} is a consequence of the following lemma.

\begin{restatable}{lem}{main}\lemlabel{main}
  The edge-maximal 3-track graph $G$ with tracks $x_1,\ldots,x_{n_1}$, $y_1,\ldots,y_{n_2}$, and $z_1,\ldots,x_{n_3}$ described above has a 2-weak layered path
  decomposition, $B_1,\ldots,B_p$, with a layering $\ell$ of (layered) pathwidth $2$ in which
  \begin{enumerate}
    \item for each $i\in\{1,2,3\}$ and each $v\in T_i$,
      $\ell(v)\equiv i\pmod 3$;
    \item $B_1=\{x_1,y_1,z_1\}$;
    \item $\ell(x_1)=1$, $\ell(y_1)=2$, and $\ell(z_1)=3$;
    \item $B_p=\{x_{n_1},y_{n_2},z_{n_3}\}$; and
    \item $x_{n_1},y_{n_2},z_{n_3}$ appear in 3 distinct consecutive layers.
  \end{enumerate}
\end{restatable}

Before proving \lemref{main}, we show how it implies \thmref{main}.
Since layered pathwidth is monotone with respect to the addition of edges,
it is safe to assume (as \lemref{main} does) that $G$ is edge-maximal.
By \lemref{main}, therefore $G$ has $\lpw_2(G)\le2$ and therefore, by
\lemref{weak}, $\lpw(G)\le 4$.

\begin{proof}[Proof of \lemref{main}]
  If $|V(G)\le 4$, then the result is trivial so assume from this point onward that $|V(G)\ge 5$.  Next, suppose $V(G)$ contains
  a cut set $C=\{x_i,y_j,z_k\}$ having exactly one vertex in each track.
  Since $G$ is edge-maximal, $x_i,y_j,z_k$ form a cycle in $G$.  Now,
  the subgraph $G_1$ of $G$ induced by $\{x_1,\ldots,x_i, y_1,\ldots,y_j,
  z_1,\ldots,z_k\}$ is an edge-maximal graph with $\tr(G_1)=3$ and we
  can inductively apply \lemref{main} to find a width-2 2-weak layered
  path decomposition of $G_1$ in which $x_i,y_j,z_k$ are in the last bag
  and are assigned to three consecutive distinct layers $r+1$, $r+2$, and $r+3$.
  Note that there are three possible assignments of $x_i,y_j,z_k$ to
  these three layers depending on the value of $r\bmod 3$.  Suppose,
  without loss of generality, that $\ell(y_j)=r+1$ (so $\ell(z_k)=r+2$
  and $\ell(x_i)=r+3$.)

  Next, consider the graph $G_2$ induced by
  $\{x_i,\ldots,x_{n_1},y_j,\ldots,y_{n_2},z_k,\ldots,z_{n_3}\}$.
  We apply \lemref{main} inductively on $G_2$ relabelling tracks to
  ensure that in the resulting layered decomposition $\ell(y_j)=1$,
  $\ell(z_k)=2$ and $\ell(x_i)=3$.   We can now obtain a width-2 2-weak
  layered path decomposition of $G$ by joining the two decompositions.
  In particular, concatenating the sequence of bags for $G_1$ with
  the sequence of bags for $G_2$ gives a path decomposition of $G$
  and adding $r$ to the indices of all layers in the layering of $G_2$
  gives a 2-weak layering of $G$.

  Thus, all that remains is to study the case where $G$ contains no cut
  set having exactly one vertex on each track.  We claim that, in this
  case, $G$ contains the edge $x_1z_2$ or it contains the edge $z_1x_2$.
  Since $G$ is edge-maximal, this is obvious unless $n_1=n_3=1$ so
  that neither $z_2$ nor $x_2$ exist.  However, since $|V(G)|\ge 5$, $n_2\ge 3$, so  this would imply that $x_1,z_1,y_2$ is a cut set with one vertex on
  each track, since its removal separates all $y_1$ from $y_3$.

  We will construct a path $P=v_1,\ldots,v_r$, an example of which is illustrated in \figref{graph-path}(a). The first vertex of $P$ will be one of
  $x_1,y_1,z_1$ and the final three vertices are $x_{n_1}$,
  $y_{n_2}$, and $z_{n_3}$, though not necessarily in that order.
  The path $P$ will \emph{spiral} in the sense that $v_i\in T_i$
  for all $i\in\{1,\ldots,r\}$.  Observe that this spiralling implies that the subsequence of vertices of $P$ on any track $T_i$ is increasing (getting further from the origin).


  \begin{figure}
  \begin{center}
  \begin{tabular}{ccc}
  \includegraphics[width=.31\textwidth]{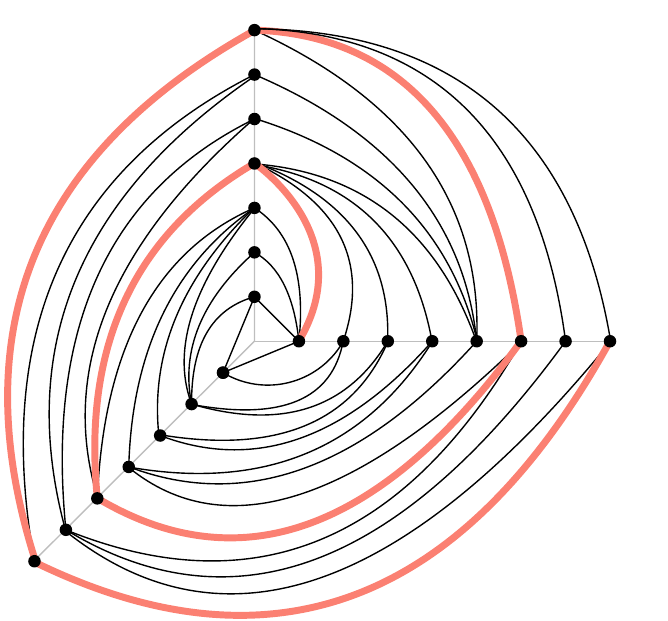} &
  \includegraphics[width=.31\textwidth]{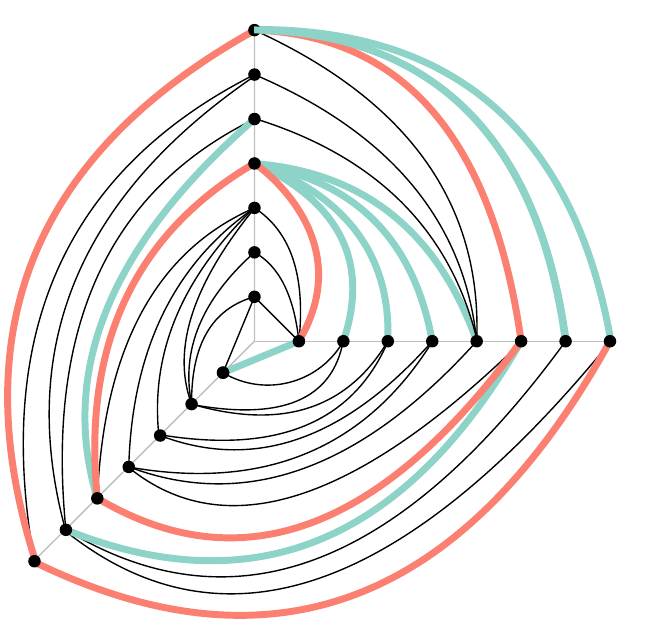} &
  \includegraphics[width=.31\textwidth]{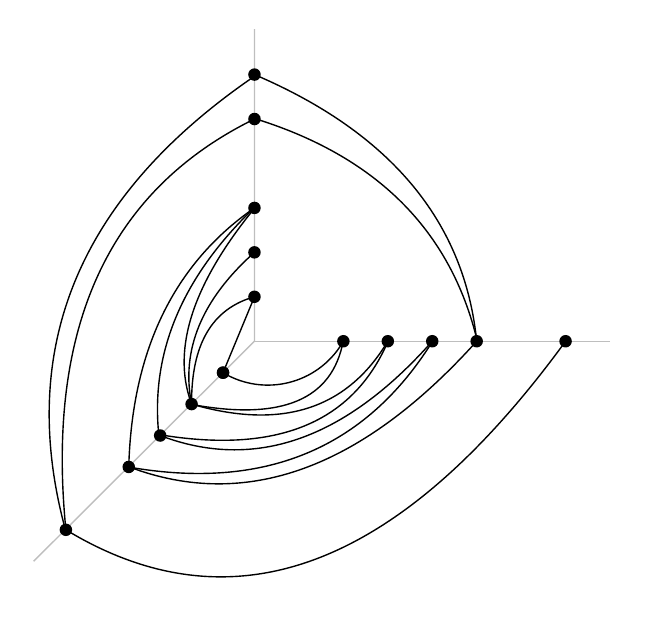} \\
  (a) & (b) & (c)
  \end{tabular}
  \end{center}
  \caption{A three track graph $G$ with (a)~the spiral curve, (b)~edges spanning two layers highlighted, (c) the levelled planar graph $G-P$.}
  \figlabel{graph-path}
  \end{figure}

  $P$ is constructed greedily: if $P$ has reached vertex $v_k$, it continues to the neighbouring vertex
  $v_{k+1}$ of $v_{k}$ with the highest index on track $T_{k+1}$ that is not yet in $P$.
  We will call this vertex $v_{k+1}$ the \emph{furthest neighbouring vertex} of $v_k$.  To see why this is always possible,
  recall that $P$ already contains an edge $v_{k-3},v_{k-2}$. Now, without loss of generality we can
  apply \obsref{silly} with $x_i=v_k$ and $y_j=v_{k-2}$, so there
  are three cases (see \figref{sloppy}):

\begin{figure}
   \begin{center}
     \begin{tabular}{ccc}
       \includegraphics[width=0.3\columnwidth]{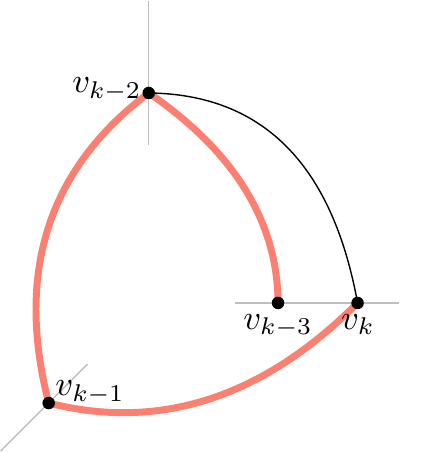} &
       \includegraphics[width=0.3\columnwidth]{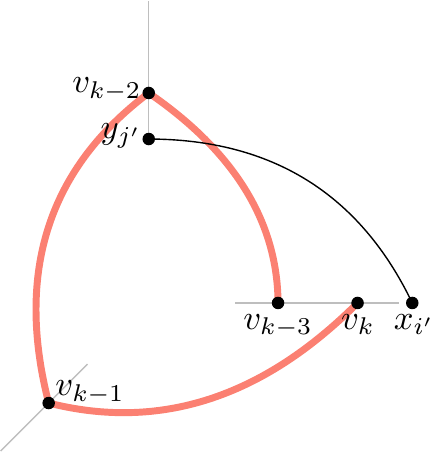} &
       \includegraphics[width=0.3\columnwidth]{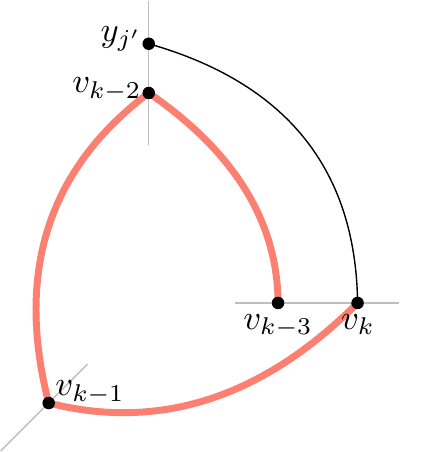} \\
       (1) & (2) & (3)
     \end{tabular}
   \end{center}
   \caption{The path $P$ can always be extended.}
   \figlabel{sloppy}
\end{figure}

  \begin{enumerate}
  \item $v_k v_{k-2}\in E(G)$.  In this
  case $v_{k-2}$, $v_{k-1}$, and $v_k$ form a cycle in $G$.  Then either
  $\{v_{k-2},v_{k-1},v_{k}\}=\{x_{n_1},y_{n_2},z_{n_3}\}$ or
  $\{v_{k-2},v_{k-1},v_{k}\}$ is a cut set with exactly one vertex in
  each track.  In the former case, the path $P$ is complete. The latter case is  excluded by the assumption that $G$ contains no such cut sets.

  \item there exists an edge $x_{i'}y_{j'}\in E(G)$ with $i' > i$
    (i.e. $i' > k$) and $j'< j$ (i.e. $j'< k-2$).  This case is not possible, since this edge would cross $v_{k-3}v_{k-2}$.

  \item there exists an edge $v_k y_{j'}\in E(G)$ with $j' >j$  (i.e. $j'> k-2$).  In this case, $P$ is extended by adding $v_{k+1}=y_{j'}$.
  \end{enumerate}

  Thus we have constructed the furthest vertex path $P=v_1,\ldots,v_r$ whose first vertex is one of
  $x_1,y_1,z_1$ and whose last three vertices are $x_{n_1}$,
  $y_{n_2}$ and $z_{n_3}$ (not necessarily in order).  We assign layers
  to the vertices of $P$ as follows: For each vertex $v_i$ on $P$,
  we set $\ell(v_i)=i$.  Note that this satisfies Conditions~3 and 5
  of the lemma and also satisfies Condition~1 for the vertices of $P$.
  For each $t\in\{1,2,3\}$, any vertex $v\in T_t$ that is not in
  $P$ is assigned to the same layer as $v$'s immediate successor in $P\cap T_t$.
  This assignment satisfies Condition~1 for vertices not in $P$.
  Finally, we will prove that this gives a 2-weak layering of $G$. In other words,
  in the worst case, a vertex $v$ with $\ell(v)=i$ can only share an edge with vertex $u$ where  $i-2 \le \ell(u) \le i+2$.  See \figref{graph-path}(b).

 Any edge between $v$ and $w$ where neither $v$ nor $w$ is in $P$ will only span one layer.
 Any edge between any two vertices $v_i$ and $v_j$ where $v_i, v_j \in P$, will span only one layer if $j = i\pm1$.
 This would mean that $v_{i}v_{j}$ is an edge in the graph $G$
 and that this edge was used to construct our furthest vertex path $P$.
 If $j \neq i\pm1$, then there are two cases:
   \begin{enumerate}
	\item $j = i \pm 2$ Such an edge is possible, and allowed since it spans only two layers. 
	\item $j = i \pm 4$ Such an edge cannot exist since it would contradict our greedy path constructing algorithm.
	If the edge $v_iv_{i+4}$ (or the edge $v_{i-4}v_i$) existed then the edge $v_i v_{i+1}$ ($v_{i-4}v_{i-3}$ ) would not have been added to $P$.
   \end{enumerate}

 Any edge between $v$ and $w$ where $v \in P$ and $w \not\in P$ will be one of 7 types (see \figref{2layermaximum}).
 Without loss of generality, assume the spiral is travelling from $T_1$ to $T_2$ to $T_3$. Let $x_i$ be a vertex on the constructed path $P$.
 First, we look at the possible cases for an edge between $x_i$ with $\ell(x_i) = m$ and $y_j$ where $y_j \notin P$.

 \renewcommand{\thesubfigure}{\arabic{subfigure}}
  \begin{figure}
    \begin{center}
    \begin{subfigure}[t]{0.3\hsize}\includegraphics[width=40mm]{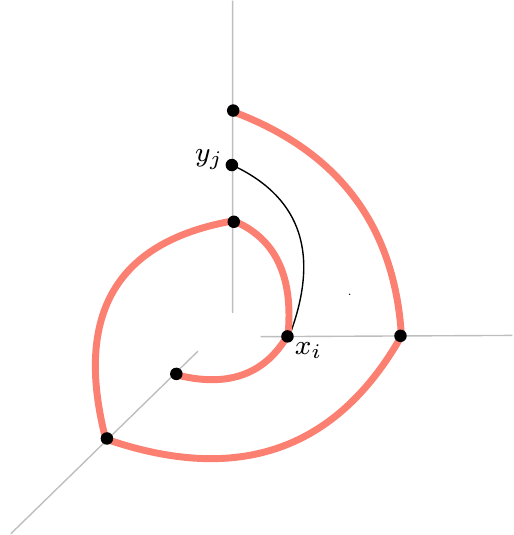} \caption{}\end{subfigure}
    \begin{subfigure}[t]{0.3\hsize}\includegraphics[width=40mm]{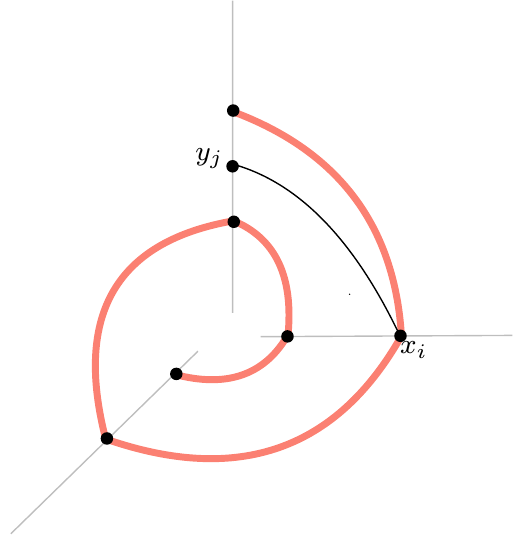} \caption{}\end{subfigure}
    \begin{subfigure}[t]{0.3\hsize}\includegraphics[width=40mm]{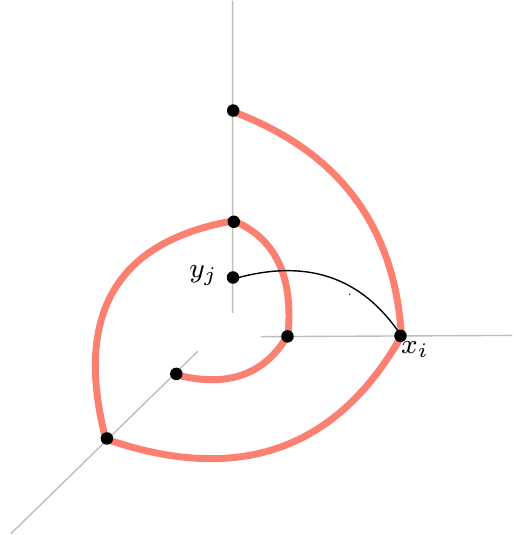} \caption{}\end{subfigure}
    \\[2ex]
    \begin{subfigure}[t]{0.4\hsize}\includegraphics[width=40mm]{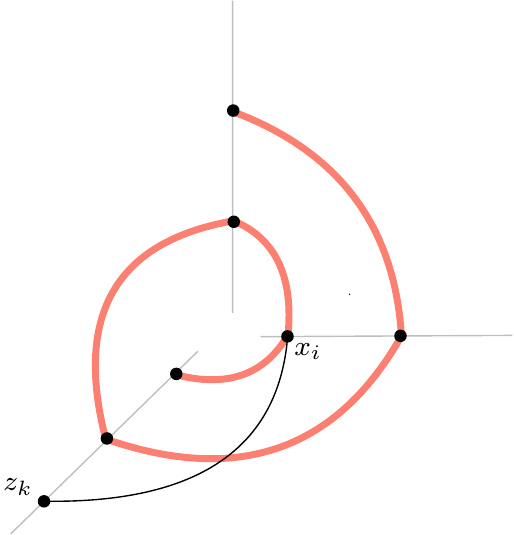} \caption{}\end{subfigure}
    \begin{subfigure}[t]{0.4\hsize}\includegraphics[width=40mm]{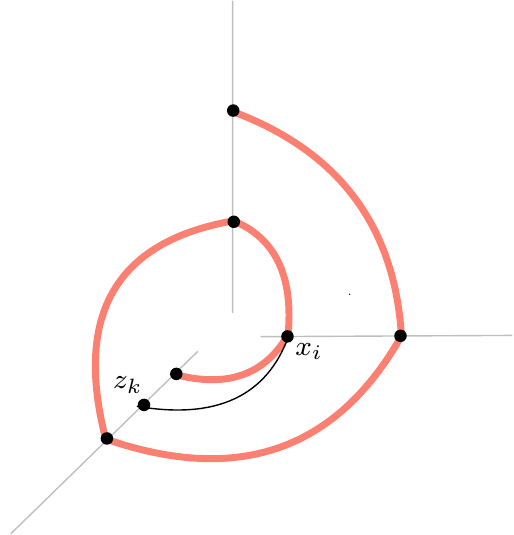} \caption{}\end{subfigure}
    \\[2ex]
    \begin{subfigure}[t]{0.4\hsize}\includegraphics[width=40mm]{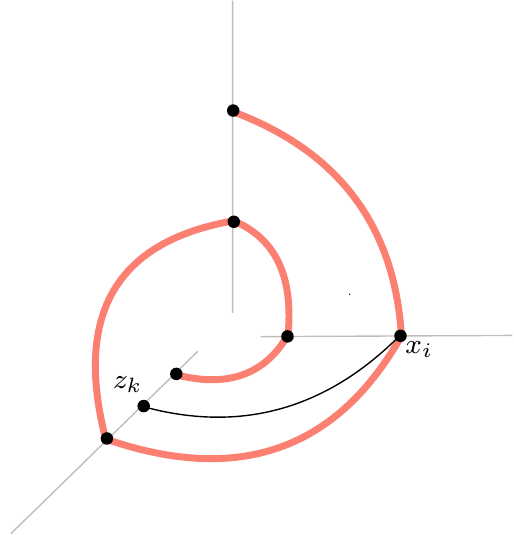} \caption{}\end{subfigure}
    \begin{subfigure}[t]{0.4\hsize}\includegraphics[width=40mm]{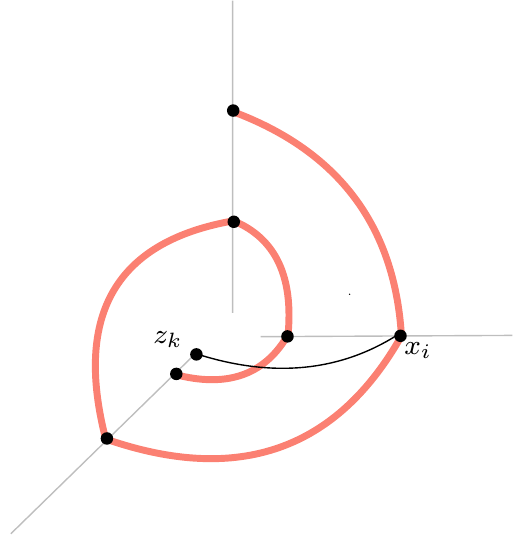} \caption{}\end{subfigure}
    \caption{The edge between a vertex $x_i$ and a vertex $y_j$ or $z_k$ cannot span more than 2 layers}
   \figlabel{2layermaximum}
 \end{center}
    \end{figure}

  \begin{enumerate}
  	\item$\ell(y_j) =  m+3n$ where $n \geq 1$. This edge cannot exist since it would contradict our greedy path constructing algorithm.
  	\item$\ell(y_j) =  m+1$. This edge will only span one layer.
  	\item $\ell(y_j) = m+1-3n$ where $n \geq 1$. This edge cannot exist, since it would create a crossing with the edge $v_{m-3}v_{m-2}$.
  \end{enumerate}

  Second, we look at the possible cases for an edge between $x_i$ with $\ell(x_i) = m$ and $z_k$ where $z_k \notin P$.
 \begin{enumerate}
  \setcounter{enumi}{3}
  	\item $\ell(z_k) = m+2+3n$ where $n \geq 1$. This edge cannot exist, since it would create a crossing with the edge $v_{m+2}v_{m+3}$.
  	\item$\ell(z_k) =  m+2$. This edge spans exactly two layers.
  	\item$\ell(z_k) =  m-1$. This edge will only span one layer.
  	\item $\ell(z_k) =  m-1-3n$ where $n \geq 1$. This edge cannot exist, since it would create a crossing with the edge $v_{m-4}v_{m-3}$.
  \end{enumerate}

  Next, we will need a notion of levelled planar graphs. The class of levelled
  planar graphs was introduced in 1992 by Heath and Rosenberg
  \cite{HR-SJC92} in their study of queue layouts of graphs. A \emph{levelled
  planar drawing} of a graph is a straight-line crossing-free drawing in
  the plane, such that the vertices are placed on a sequence of parallel
  lines (called levels), where each edge joins vertices in two
  consecutive levels. A graph is \emph{levelled planar} if it has a levelled
  planar drawing. (This is a well studied model for planar
  graph drawing, known as a Sugiyama-style drawing \cite{STT81,BM2001,HN2013,BETT99}.)

  Now, consider the graph $G-P$ obtained by removing the vertices of $P$
  from $G$ (see \figref{graph-path}(c)).  We claim that this graph is a levelled
  planar graph in which the levels
  of the vertices are given by the layering $\ell$ defined above. (Recall that the edges of $G$ spanning two layers all have at least one endpoint in $P$.)
  Refer to \figref{prism}. One way to see this is to imagine $G$
  as being drawn with its vertices on the three vertical edges of the surface of
  a triangular prism so that $x_1,y_1,z_1$ are the vertices of one
  triangular face and $x_{n_1},y_{n_2},z_{n_3}$ are the vertices of
  the other triangular face.  Now, if we remove the triangular faces
  of this prism, cut it along the embedding of $P$, and unfold
  the resulting surface so that it lies in the plane, then we obtain a
  drawing of $G-P$ in which the vertices lie on a set of parallel lines
  and in which the edges only join vertices on two consecutive lines.
  This gives the desired levelled planar drawing of $G-P$.


  \begin{figure}
   \begin{center}
      \includegraphics[width=0.95\columnwidth]{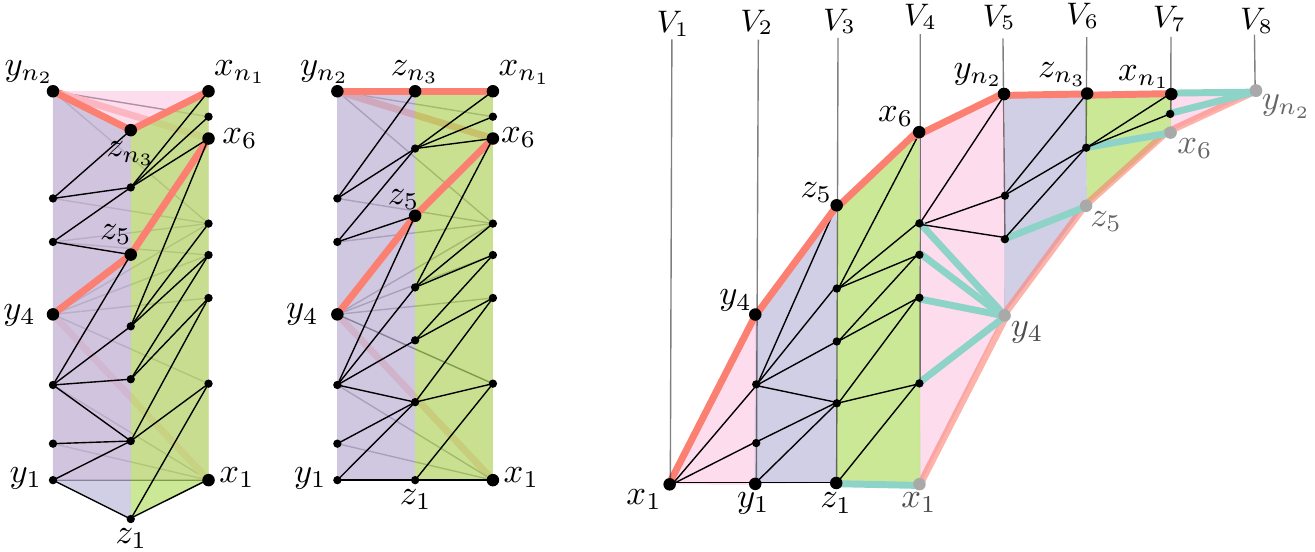}
   \end{center}
   \caption{Cutting a prism along $P$ to obtain a levelled planar
    drawing of $G-P$.}
   \figlabel{prism}
  \end{figure}


  By a result of Bannister \etal\ \cite[Proof of
  Theorem~5]{bannister2018track}, $G-P$ has a layered
  path decomposition $B_1,\ldots,B_p$ of width 1 using the layering
  $\ell$ defined above.  If we add the vertices of $P$ to every bag
  of this decomposition we obtain a width-2 2-weak layered path
  decomposition of $G$.  Finally, to satisfy Conditions 2 and 4 of
  the lemma, we prepend a bag $B_0=\{x_1,y_1,z_1\}$ and append a bag
  $B_{p+1}=\{x_{n_1},y_{n_2},z_{n_3}\}$.
\end{proof}

\section{Conclusions}

We have presented two results on layered pathwidth that help complete our understanding of the relationship between layered pathwidth, stack number, and track number.  Graphs of bounded layered pathwidth include grids and certain types of intersection graphs.  \thmref{stacknumber}, for example, implies that unit-disk graphs with maximum clique size $k$ have stack number $O(k)$ since they have been shown to have $O(k)$ layered pathwidth \cite{DBLP:conf/gd/BannisterDDEW16,bannister2018track}.

We conclude this discussion by remarking that upper bounds on the stack number of graphs of bounded layered \emph{treewidth} are not yet known, and present a challenging avenue for further study.  For example, $k$-planar graphs are known to have layered treewidth $O(k)$ \cite[Theorem~3.1]{DBLP:journals/siamdm/DujmovicEW17}.  Therefore, bounding stack number by a function of layered treewidth would imply that $k$-planar graphs have bounded stack-number. It is still unknown whether $k$-planar graphs have bounded stack-number except in the case $k=1$ \cite{DBLP:journals/algorithmica/BekosBKR17,DBLP:journals/corr/AlamBK15}.

\bibliographystyle{plain}
\bibliography{lpw}

\end{document}